\documentclass[runningheads,a4paper]{llncs}
\usepackage[margin=1.33in]{geometry}
\usepackage{amssymb}
\usepackage{graphicx}
\usepackage{verbatim}
\usepackage{mathrsfs}
\usepackage[ruled,vlined,lined,commentsnumbered,linesnumbered]{algorithm2e}
\usepackage{amsfonts}

\usepackage{amsthm}
\usepackage{ifpdf}
\usepackage{amsmath}
\usepackage[noend]{algpseudocode}
\usepackage{enumerate}
\usepackage{url}
\usepackage{color,soul}
\usepackage{hyperref}
\usepackage{cite}

\urldef{\mailsa}\path|{ibanerje, richards}@cs.gmu.edu|    
\newcommand{\keywords}[1]{\par\addvspace\baselineskip
\noindent\keywordname\enspace\ignorespaces#1}

\begin{document}

\mainmatter  

\title{Routing Number Of A Pyramid}

\titlerunning{Routing Number Of A Pyramid}

%
%
\author{Indranil Banerjee, Dana Richards}
\authorrunning{I  Banerjee, D Richards}

\institute{George Mason University\\ Department Of Computer Science\\ Fairfax Virginia 22030, USA\\
\mailsa}

%
%

\toctitle{Lecture Notes in Computer Science}
\tocauthor{Authors' Instructions}
\maketitle

\begin{abstract}
In this short note we give the routing number of pyramid graph under the \textit{routing via matching} model introduced by Alon et al\cite{5}. This model can be viewed as a communication scheme on a distributed network. The nodes in the network can communicate via matchings (a step), where a node exchanges data with its partner. Formally, given a connected graph $G$ with vertices labeled from $[1,...,n]$ and a permutation $\pi$ giving the destination of pebbles on the vertices the problem is to find a minimum step routing scheme. This is denoted as the routing time $rt(G,\pi)$ of $G$ given $\pi$. We show that a $d$-dimensional pyramid with $m$ levels has a routing number of $O(dN^{1/d})$.

\keywords{Routing, Pyramids, Permutations}
\end{abstract}

\section{Routing Number Of A Pyramid}
Originally introduced by Alon and others \cite{5} this problems explores permutation routing on graphs where routing is achieve through a  series of matchings called steps. Let $G$ be an undirected labeled graph with vertex labeled $i$ starting with a pebble labeled $\pi_i$, and the permutation $\pi$ gives the destinations of each pebble. The task is to route each pebble to their destination via a sequence of matchings. Given a matching we swap the pebbles on matched vertices. The \textit{routing time} $rt(G,\pi)$ is defined as the minimum number of steps, consecutive matchings, necessary to route all the pebbles for the given permutation. For given graph $G$, the maximum routing time over all permutations is called  the routing number $rt(G)$ of $G$. Permutation routing via matching has generated a literature that has focused on determining the routing numbers of special graphs such as trees, cycles, hypercubes etc. \cite{1,2,3,4}. In this short note we add to this literature by determining the routing number of the pyramid graph. As a consequence of our proof technique we also obtain the routing number of a multi-grid\cite{5}.

\
\begin{figure}[h]
	\includegraphics[width=4.5cm]{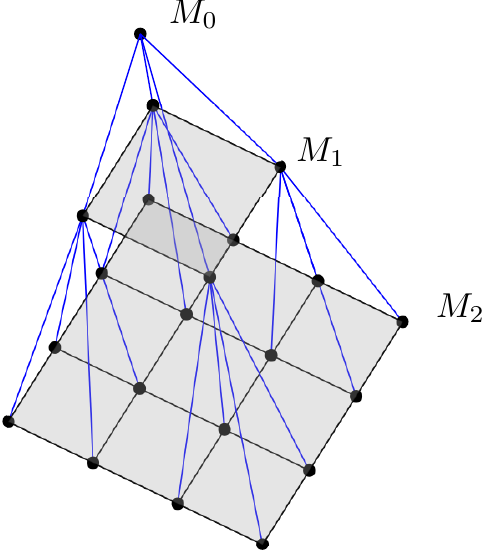}
	\centering
	\caption{A pyramid $\triangle_{3,2}$ in 3-dimension}
\end{figure}

A 1-dimensional pyramid with $m$-levels is defined as the complete binary tree of $2^{m} -1$ nodes, where the nodes in each level is connected by a path (i.e., a one-dimensional mesh). We treat the apex (root) to be at level 0 and subsequent levels are numbered in ascending order. A 2-dimensional  pyramid is shown in Figure 1. In this case each level $l$ is a square mesh of size $4^{l}$. Similarly a $d$-dimensional  pyramid having $m$ levels is denoted by $\triangle_{m,d}$ where the level $l$ is a $d$-dimensional regular mesh of length $2^{l}$ in each dimension. Clearly, $N = |\triangle_{m,d}|= \frac{2^{md}-1}{2^{d}-1}$ and the size of level numbered $l$ is $n_l = 2^{dl}$.  

\begin{theorem}\label{lemma: rt pyramid}
	For any pyramid $\triangle_{m,d}$, $rt(\triangle_{m, d}) = O(dN^{1/d})$.
\end{theorem} 

\begin{proof}
	\begin{figure}[h]
		\includegraphics[width=4.5cm]{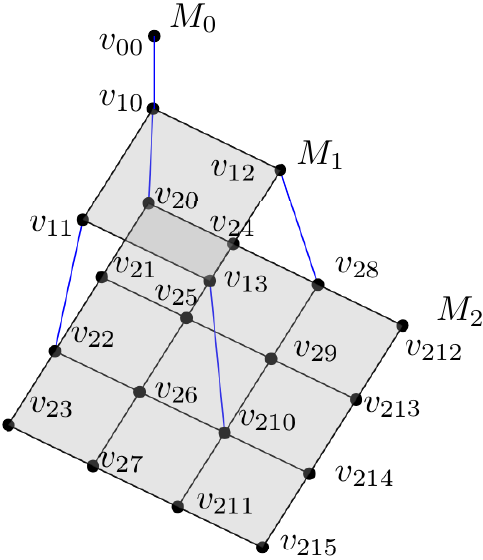}
		\centering
		\caption{The graph $\triangle'_{3,2}$ after stripping way edges from $\triangle_{3,2}$}
	\end{figure}
	
	Given the pyramid $\triangle_{3,2}$ consider a subgraph $\triangle'_{3,2}$ as shown in Figure 2. In literature this graph is sometimes refer to as a multi-grid, see for example \cite{5}.  As we move down from the apex we remove all but the ``first'' edge from the set of edges that connects a vertex to its neighbors in the level below. The remaining edges that connects two adjacent layers will be referred to as vertical edges. These edges can be grouped into disjoint vertical paths as shown by the blue lines in Figure 5.  The above construction naturally generalizes in higher dimensions. Clearly $rt(\triangle_{m,d}) \le rt(\triangle'_{m,d})$ where $ \triangle'_{m,d} $ is the multi-grid obtained from $ \triangle_{m,d} $.  We shall show $rt(\triangle'_{m,d}) = O(dN^{1/d})$.  
	
	Let $\pi$ be some input permutation. Without loss of generality we assume that $\pi$ consists only of 2-cycles or 1-cycles. From \cite{5} we know that any arbitrary permutation can be written as a composition of at most two such permutations. In order to route $\pi$ we first route the pebbles into their appropriate levels and then route within these levels. Routing consists of five rounds where in the odd numbered rounds we route within the levels and in the even numbered rounds we use the vertical paths to route between the levels. The first four rounds are used to move the pebbles to their appropriate destination level.
	
	 Let $v_{ij}$ be the $j^{th}$ node at level $i$, where $j \in [0, n_i-1]$. Let $\phi_k$ be the number of maximal vertical paths of length $k$. For example, in Figure 5, $\phi_2=1$ and $\phi_1=3$. In general in a $\triangle'_{m,d}$, $\phi_k = n_{m-k-1}-n_{m-k-2}$ for $k \in [1,m-2]$ and $\phi_{m-1}=1$. We group the cycles in $\pi$ based on their source and destination level (in case of 1-cycles the source and destination levels are the same). Let $P_{ij}$ ($i < j$) be the set of pebble pairs that need to be moved from level $i$ down to level $j$ and vice-versa and $P_{ii}$ be the set of pebbles that stay in level $i$. Let $\mu_{ij} = |P_{ij}|$. Let $P_i=\bigcup_{i<j}P_{ij}$ be the set of pebble pairs that move a pebble up to level $i$.
We shall only use disjoint vertical paths of length $m-i-1$ to route the pebbles in $P_{i}$.
During either even round each path of length $m-i-1$ will be used, for some $j$, to swap two pebbles between levels $i$ and $j$; all other pebbles on that path will not move. As an example consider the case in Figure 2. Suppose $\pi(v_{00}) = v_{21}$. Then during the intra-level routing on the first round we will move the pebble at $v_{21}$ to $v_{20}$. All intermediate nodes on this path, which in this case is just $v_{10}$ will be ignored (i.e., a pebble on these nodes will return to their original position at the end of the round). The four pebbles $\{v_{10},v_{11},v_{12},v_{13}\}$ will only use the three paths of length 1 to move to the bottom level (if necessary). In general $|P_i|=\sum_{j > i}{\mu_{ij}} \le n_i \le 2(n_{i}-n_{i-1}) = 2\phi_{m-i-1}$. Hence we need at most two rounds of routing along these vertical paths to move all pebbles in $P_i$.  
	 
	 Routing within the levels (which happens in parallel) is dominated by the routing number of the last level which is known to be $O(dn_{m}^{1/d}) = O(dN^{1/d})$ (for example we can use corollary 2 of theorem 4 in \cite{1}). Hence the three odd rounds take $O(dN^{1/d})$ in total. In the even rounds routing happens in parallel along the disjoint vertical paths. The routing time in this case is $O(m)$. Since, $N^{1/d} = \Omega(2^m)$, the even rounds do not contribute to the overall routing time, which remains $O(dN^{1/d})$, as claimed by the theorem.
	
\end{proof}

%
%
%

%
%
%


\end{document}